\documentclass[aps,prx,reprint,superscriptaddress,amsmath,amssymb]{revtex4-2}

\usepackage{graphicx}
\usepackage{dcolumn}
\usepackage{bm}
\usepackage{latexsym}
\usepackage{epsfig}
\usepackage{amsbsy}
\usepackage{array}
\usepackage{tabularray}
\usepackage{setspace}
\usepackage{physics}
\usepackage{comment}
\usepackage{multirow}
\usepackage{natbib}
\usepackage{units}
\usepackage{float}

\usepackage{amsmath,amssymb,amsthm,yhmath,amsfonts,mathrsfs,enumitem,mathtools,dsfont}
\usepackage[utf8]{inputenc}
\usepackage{tikz}
\usetikzlibrary{shapes,arrows,positioning,automata}

\newtheorem{theorem}{Theorem}

\newtheorem{remark}{Remark}

\newlength{\tblwidth}
\usepackage{apptools}
\AtAppendix{\counterwithin{theorem}{section}}

\usepackage{algpseudocode}

\usepackage[justification=centerlast]{caption}
\usepackage[labelformat=simple]{subcaption}

\usepackage[justification=centerlast]{caption}

\usepackage{xr}
\begin{document}

\title{Preferential Attachment with Local Flexibility}

\author{Tingyu Zhao}
\thanks{These authors contributed equally.}
\affiliation{Department of Industrial Engineering and Management Sciences, Northwestern University, Evanston, Illinois 60208, USA}
\affiliation{NSF-Simons National Institute for Theory and Mathematics in Biology, Chicago, Illinois 60611, USA}

\author{Balázs~Maga}
\thanks{These authors contributed equally.}
\affiliation{HUN-REN Alfr\'ed R\'enyi Institute of Mathematics, Budapest 1053, Hungary}

\author{Pierfrancesco Dionigi}
\affiliation{HUN-REN Alfr\'ed R\'enyi Institute of Mathematics, Budapest 1053, Hungary}

\author{Gergely Ódor}
\affiliation{HUN-REN Alfr\'ed R\'enyi Institute of Mathematics, Budapest 1053, Hungary}
\affiliation{Institute for Hygiene and Applied Immunology, Center for Pathophysiology, Infectiology and Immunology, Medical University of Vienna, Vienna 1090, Austria}
\affiliation{Department of Network and Data Science, Central European University, Vienna 1100, Austria}

\author{Kyle Soni}
\affiliation{Department of Physics and Astronomy, Northwestern University, Evanston, Illinois 60208, USA}

\author{Anastasiya Salova}
\affiliation{Department of Physics and Astronomy, Northwestern University, Evanston, Illinois 60208, USA}
\affiliation{Department of Engineering Sciences and Applied Mathematics, Northwestern University, Evanston, Illinois 60208, USA}

\author{Bingjie Hao}
\affiliation{Department of Physics and Astronomy, Northwestern University, Evanston, Illinois 60208, USA}

\author{Miklós~Abért}
\email{karinthy@gmail.com}
\affiliation{HUN-REN Alfr\'ed R\'enyi Institute of Mathematics, Budapest 1053, Hungary}

\author{István~A.~Kovács}
\email{istvan.kovacs@northwestern.edu}
\affiliation{NSF-Simons National Institute for Theory and Mathematics in Biology, Chicago, Illinois 60611, USA}
\affiliation{Department of Physics and Astronomy, Northwestern University, Evanston, Illinois 60208, USA}
\affiliation{Department of Engineering Sciences and Applied Mathematics, Northwestern University, Evanston, Illinois 60208, USA}
\affiliation{Northwestern Institute on Complex Systems, Northwestern University, Evanston, Illinois 60208, USA}

\date{\today}

\begin{abstract}
From the formation of social ties to the budding quantum internet, growing networks often exhibit local flexibility upon new nodes attaching to an existing network.
In our proposed model, a new node connects uniformly at random to a node within the proximity of the intended target, including, but not restricted to, the target itself. 
Through numerical simulations and rigorous stochastic analysis, we find this local flexibility to qualitatively change the global network behavior of nonlinear preferential attachment.
Depending on whether the preferential attachment is superlinear or (sub)linear, two distinct classes of complex network architectures emerge. The superlinear phase leads to a layered hierarchy, with no stationary degree distribution. Although there is a stationary degree distribution in the linear and sublinear cases, it decays strictly faster than for the Barabási--Albert model.
We interpret our results within a two-dimensional phase diagram of network growth models incorporating redirection, with broad implications. 
\end{abstract}

\maketitle

In the advancement of network science, 
minimal generative models have played a central role by providing mechanistic insights into the emergence of complex network topologies.
Variants of preferential attachment have been especially influential ~\cite{barabasi1999emergence, krapivsky2000connectivity}. In these models, incoming nodes in a growing network attach to existing nodes stochastically, with attachment probabilities dictated by node popularity, modeled as a function of node degree---the number of current partners. In particular, under linear preferential attachment when popularity is directly given by the degree, known as the Barabási--Albert (BA) model~\cite{barabasi1999emergence}, networks self-organize into a scale-free structure, characterized by a power-law degree distribution~\cite{newman2005power}.
While preferential attachment has widely served as a benchmark for understanding how complex networks emerge, it has also motivated a multitude of generalizations and alternative frameworks~\cite{dorogovtsev2002evolution}, incorporating growth factors such as intrinsic node fitness~\cite{bianconi2001competition, bianconi2001bose, caldarelli2002scale}, node aging~\cite{dorogovtsev2000evolution}, node similarity~\cite{papadopoulos2012popularity}, node centrality~\cite{topirceanu2018weighted, adami2024dandelion}, or node copying~\cite{kumar2000stochastic, vazquez2003modeling, krapivsky2005network}.

An implicit, and often unjustified, assumption in preferential attachment models is that a new node always connects to its chosen target \emph{directly} and \emph{precisely}. In reality, however, dynamics in complex systems are often subject to \emph{local attachment flexibility}: in growing networks, a new node guided toward a target may be redirected, establishing a link to a nearby node instead.
For example, Alice may approach Bob for a social connection---friendship, academic collaboration, a business opportunity, or otherwise---but may eventually build a relationship with Carol, one of Bob's contacts, instead.
This could happen because Bob refers Alice to Carol as an act of altruism~\cite{li2011emergence}; alternatively, Alice may have diffused away, exploring Bob's neighborhood in the spirit of a random walker~\cite{saramaki2004scale, evans2005scale, cannings2013random, krapivsky2017emergent, berry2024random}.
In either case, Carol's attractiveness no longer derives solely from her own popularity, but also from Bob's popularity, reflecting \emph{indirect utility} through a node's neighborhood~\cite{lee2023emergence}.
Local flexibility is also consistent with celebrated social theories such as the ``strength of weak ties''~\cite{granovetter1973strength} and structural holes~\cite{burt2004structural}, which highlight that the formation of social ties is driven not only by individual popularity, but also by the surrounding local network context.

Local flexibility is inherently at play beyond social scenarios. In biological systems, limited control and imperfect information often result in local variability in link formation~\cite{shahrezaei2008stochastic}. 
For example, in developing brains, neurons form synapses under the guidance of chemical gradients that are intrinsically noisy; moreover, the genome lacks the capacity to encode wiring rules for every synapse, necessitating coarse-grained developmental mechanisms that naturally produce flexible connectivity~\cite{tessier1996molecular, hassan2015beyond, zador2019critique}. 
Finally, our work is also inspired by nascent quantum communication networks~\cite{gisin2007quantum, kimble2008quantum}.
Even without a direct physical link between two parties, fully secure quantum communication is achievable through entanglement swapping protocols, whereby intermediate relay nodes can be used to establish an end-to-end communication channel between indirectly connected parties~\cite{zukowski1993event, riebe2008deterministic, meng2025path}. Thus, a new node introduced into the network does not need to directly connect to its intended communication target; instead, any connection into the proximity of the target suffices.
Taken together, these examples illustrate that attachment flexibility can be a generic feature of many real-world networks.

In this Letter, we investigate how local attachment flexibility in growing networks impacts the resulting global network features.
To focus on this effect in its purest form, we introduce a minimal network growth model that extends standard nonlinear preferential attachment~\cite{krapivsky2000connectivity, krapivsky2001organization} by explicitly incorporating attachment flexibility.
As our key result, such attachment flexibility leads to complex network architectures that are qualitatively distinct from preferential attachment baselines, and in particular, it \emph{never} produces BA networks.
These findings highlight the counterintuitive consequences of local perturbations in network dynamics on the global network structure.
Our results are supported by numerical simulations, analytic arguments, and, where feasible, rigorous proofs, as listed in the Appendix with detailed proofs presented in the Supplemental Material~\cite{SM}.

\begin{figure}[t!]
    \centering
    \includegraphics[width=0.75\columnwidth]{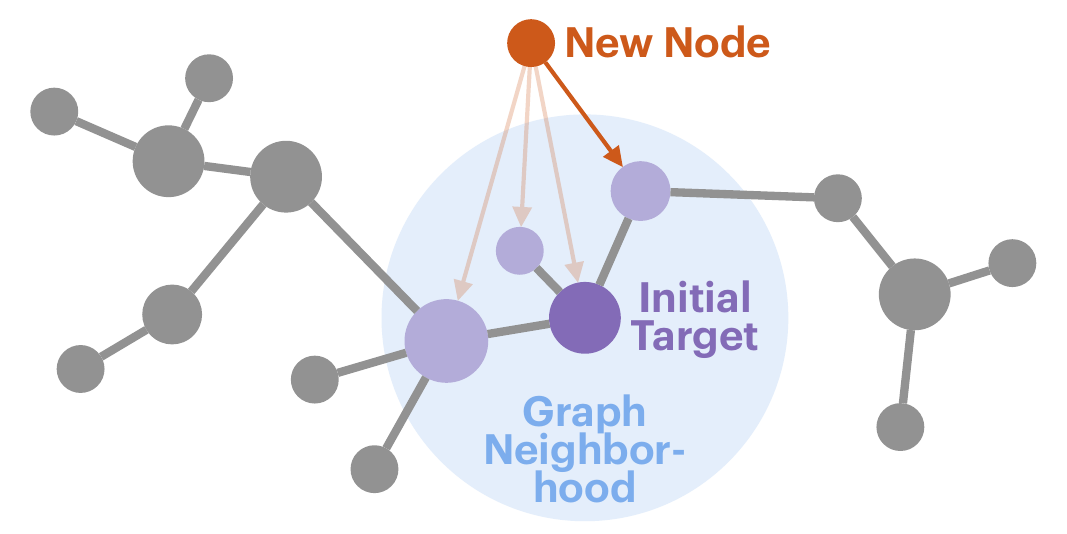}
    \vspace{-2mm}   
 \caption{\textbf{The locally flexible preferential attachment (LFPA) mechanism.}
 At each iteration, a new node selects an initial target from the existing network via (nonlinear) preferential attachment. The new node then attaches uniformly at random to any member of the target’s immediate graph neighborhood, consisting of the target and its neighbors. \hfill\hfill
 }
 \label{fig:model}
\end{figure}

\begin{figure}[t!]
    \centering
    \includegraphics[width=\columnwidth]{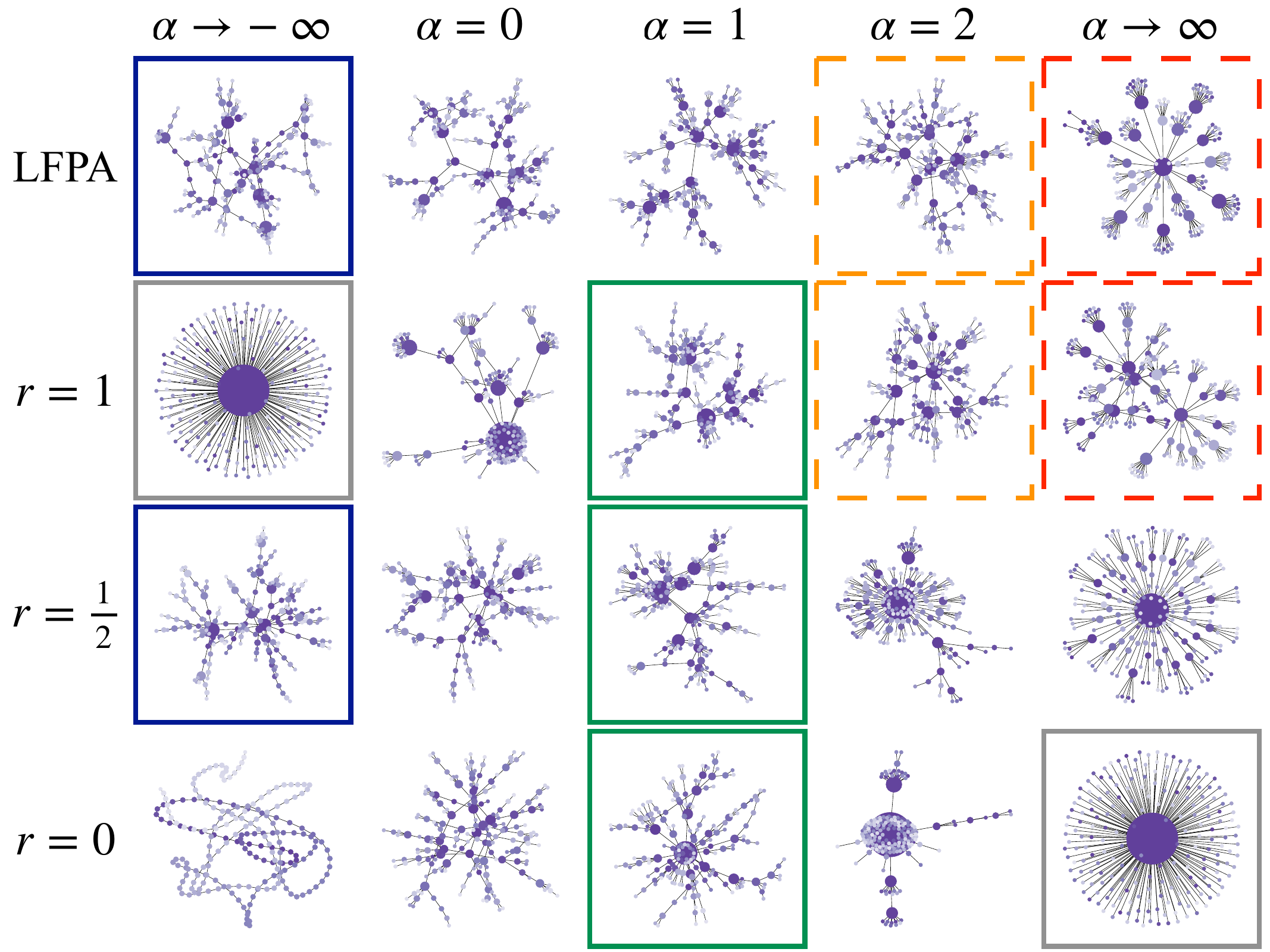}
    \vspace{-4mm}
 \caption{\textbf{Phase diagram of the studied models.} Incoming nodes first identify an initial target based on its degree $d_i$, proportionally to $d_i^\alpha$.
 A subsequent redirection happens with degree-dependent probability $r_i=d_i/(d_i+1)$ in the LFPA model, or with fixed probability $r\in[0,1]$ in the constant redirection (CR) model.
 Networks are visualized at $200$ nodes, starting from three nodes in a chain. Larger node size reflects higher degree and darker node shade reflects earlier entry into the network. Networks boxed in the same color indicate  either exact (solid) or conjectured asymptotic (dashed) model equivalence.  \hfill\hfill}
 \label{fig:pd} 
\end{figure}

\emph{Models and analyses.---}
We propose the \emph{locally flexible preferential attachment} (LFPA) model, parameterized by a single variable, $\alpha$, as follows. Each new node first enters the network following the standard nonlinear preferential attachment rule, selecting an initial target node $i$ with probability proportional to $d_i^\alpha$~\cite{krapivsky2000connectivity, krapivsky2001organization}, where $d_i$ is the current degree of node $i$ and $\alpha \in \mathbb{R}$ controls the nonlinearity. Then, as the simplest manifestation of local attachment flexibility, the new node may connect uniformly at random to any node in the graph neighborhood of radius $R$ of the initial target. The graph neighborhood is defined as all nodes that are reachable from the target node by performing $\leq R$ hops, including the target. $R=0$ corresponds to standard nonlinear preferential attachment with precise connections. Moving beyond this case, under the most local setting that permits local flexibility, $R = 1$, the new node either attaches to the target itself or to one of its neighbors, choosing uniformly among them, hence each having probability $1/(d_i + 1)$ of being chosen, see~Fig.~\ref{fig:model}. The new node is therefore redirected from the initial target with probability $r_i=d_i/(d_i + 1)$ toward one of its neighbors.

Here, we focus on the $R = 1$ setting for its analytic simplicity and show that it already yields fundamentally different network structures compared to $R=0$. 
We allow each node to connect only to a single existing node, resulting in undirected trees, but it is straightforward to generalize the analysis to the case where new nodes form multiple connections.
We define the \emph{weight} of node $i$ as the probability that a new node will connect to node $i$, partially due to its neighbors $j \sim i$, given by
\begin{align} \label{eq_LFPA_weights}
w_i(\alpha)
= (1-r_i)d_i^\alpha
    + \sum_{j \sim i} r_j d_j^{\alpha-1},
\end{align}
up to normalization by the total graph weight $W\equiv\sum_iw_i=\sum_{i} d_i^\alpha$.
Note that in the LFPA model, the redirection probability $r_i = d_i/(d_i+1) \in[\nicefrac{1}{2},1)$ is heterogeneous, increasing monotonically with the degree of the target node $d_i$. To facilitate further analysis of LFPA, it is useful to introduce the closely related, but simpler \emph{constant redirection} (CR) model, where $r_i$ is replaced by a fixed redirection probability $r \in [0,1]$, leaving the form of Eq.~\eqref{eq_LFPA_weights} and $W$ unchanged. 
As we will show, the two-parameter $(\alpha, r)$ continuum of the CR model encapsulates a wide range of complex network topologies, some previously studied in the literature, others introduced here.
Illustrations of networks generated by the LFPA and CR models at representative values of $\alpha$ and $r$ are shown in Fig.~\ref{fig:pd}, with degree distributions plotted in Fig.~\ref{fig:degdist_comparison}. Results for key network observables are summarized in Table~\ref{tab:results} and visualized in Fig.~\ref{fig:num_obs}, synthesizing our analytic and numerical contributions.
As a function of $\alpha$, we find that the LFPA model showcases two qualitatively distinct regimes:
a superlinear phase for $\alpha > 1$ with a hierarchical structure lacking a stationary degree distribution, and a (sub)linear phase for $\alpha \le 1$ with a stationary, scale-free degree distribution whose degree exponent diverges at $\alpha \to -\infty$.

\begin{table*}[htbp!]
\begin{center}
\caption{\textbf{Overview of our results.} Analytic results derived in this work are highlighted in bold, with supporting theorems in the Appendix (with proofs in the Supplemental Material) referenced.
Conjectures based on numerical results are marked with (*). 
For the $r = \nicefrac{1}{2}$ row, we report the general result for $r \in (0,1)$ when available (†); for the $\alpha = 2$ column, we report the general result for $\alpha \in (1, \infty)$ when available (‡). 
Big-$O$ denotes asymptotic upper bounds, and big-$\Theta$ denotes matching asymptotic upper and lower bounds, i.e.~scaling up to constant multiplicative factors.
$N$ is the number of nodes in the network, and $\gamma$ stands for the power-law exponent when a network is scale-free (s.f.). 
Our framework incorporates several known models from previous literature:  $\alpha = 1, r\in [0, 1]$ is the Barabási-Albert (BA) model~\cite{barabasi1999emergence}; $\alpha = 0, r = 0$ is the random recursive tree (RRT)~\cite{frieze2015recursive}; $\alpha = 0, r = 1$ is the random friend tree (RFT)~\cite{cannings2013random, krapivsky2017emergent, berry2024random}; $\alpha = 2, r = 1$ is the ``$k2$'' model~\cite{falkenberg2020identifying, lee2023emergence}. 
\hfill\hfill}
\vskip-2mm
\small
\resizebox{\tblwidth}{!}{
\begin{tblr}{
  colspec = {
    l
    l
    Q[c,m,wd=0.22\tblwidth]
    Q[c,m,wd=0.20\tblwidth]
    Q[c,m,wd=0.18\tblwidth]
    Q[c,m,wd=0.16\tblwidth]
    Q[c,m,wd=0.16\tblwidth]
  },
  rowsep = 1pt,
  stretch = 1,
}
\hline
\hline
\phantom{XX} & Property &
  $\alpha\to-\infty$ & $\alpha=0$ & $\alpha=1$ & $\alpha=2$ & $\alpha\to\infty$ \\
\hline
\SetCell[r=4]{c,m}\rotatebox[origin=c]{90}{LFPA}
  & architecture & Weibull-like, not s.f.$^{\ref{thm:master_not_sf}}$ & s.f., $\gamma > 3$* & LFBA, s.f., $\gamma > 3^{\ref{thm:uncorrelated_scalefree_unstable}}$* & $k2$, hierarchy & Rich club \\
& $d_{\max}$ & $\Theta((\log N)^{\nicefrac{1}{k}}), k \approx 0.335$* & $\Theta(N^{\nicefrac{1}{(\gamma - 1)}})$* & $\Theta(N^{\nicefrac{1}{(\gamma - 1)}})$* &
  $\boldsymbol{\Theta(N^{\nicefrac{\alpha}{(2\alpha - 1)}})}^\ddagger$ & $\boldsymbol{\Theta(\sqrt{N})}$ \\
& leaves & ${\nicefrac{1}{2}N}$ & $\Theta(N)$* & $\Theta(N)$* & $\Theta(N)^\ddagger$* & $\Theta(N)$* \\
& diameter & $O(\log N)$* & $O(\log N)$* & $O(\log N)$* & $O(\log N)$* & $O(\log N)$* \\
\hline
\SetCell[r=4]{c,m}\rotatebox[origin=c]{90}{\shortstack{CR\\$r=1$}}
  & architecture & Star & RFT &
    BA, s.f., $\gamma=3$ & $k2$, hierarchy~\cite{lee2023emergence} & Rich club \\
& $d_{\max}$ & $N-1$ & $\Theta(N)$~\cite{berry2024random} & $\Theta(\sqrt{N})$ &
  $\boldsymbol{\Theta(N^{\nicefrac{\alpha}{(2\alpha - 1)}})}^{\ddagger, \ref{thm:layered_hierarchy}}$ &
  $\boldsymbol{\Theta(\sqrt{N})}^{\ref{thm:largest_degree_alpha=infty_r=1}, \ref{thm:layered_hierarchy}}$ \\
& leaves & $N-1$ & $ N - O( N^{0.9} )$~\cite{berry2024random} & $\nicefrac{2}{3}N$ &
  $\Theta(N)^\ddagger$* & $\Theta(N)$* \\
& diameter & $2$ & $\Theta(\log N)$~\cite{berry2024random} & $\Theta(\log N)$ &
  $O(\log N)$* & $O(\log N)$* \\
\hline
\SetCell[r=4]{c,m}\rotatebox[origin=c]{90}{\shortstack{CR\\$r=\nicefrac{1}{2}$}}
  & architecture & Weibull-like, not s.f.$^{\ref{thm:master_not_sf}}$ & s.f., $\gamma > 3$* & BA, s.f., $\gamma=3$ & Hub and spokes & Star-like \\
& $d_{\max}$ & $\Theta((\log N)^{\nicefrac{1}{k}}), k \approx 0.335$* & $\Theta(N^{\nicefrac{1}{(\gamma - 1)}})$* & $\Theta(\sqrt{N})^\dagger$ &
  $\Theta(N)$* & $\boldsymbol{\Theta(N)}^{\dagger, \ref{thm:king_alpha_infty}}$ \\
& leaves & $\boldsymbol{rN}^{\dagger, \ref{thm:leaves_alpha=-infty}}$ &
  $\boldsymbol{\Theta(N)}^{\dagger, \ref{thm:leaves_alpha=0}}$ &
  $\nicefrac{2}{3}N^\dagger$ & $\Theta(N)$* &
  $\boldsymbol{(1-r+r^2)N}^{\dagger, \ref{thm:leaves_proportion}}$ \\
& diameter & $O(\log N)$* & $O(\log N)$* & $\Theta(\log N)^\dagger$ &
  $\Theta(1)$* & $\boldsymbol{\Theta(1)}^{\dagger, \ref{thm:king_alpha_infty}}$ \\
\hline
\SetCell[r=4]{c,m}\rotatebox[origin=c]{90}{\shortstack{CR\\$r=0$}}
  & architecture & Chain & RRT, exponential~\cite{frieze2015recursive} & BA, s.f., $\gamma=3$~\cite{barabasi1999emergence, bollobas2001degree} &
    Hub and spokes & Star \\
& $d_{\max}$ & $2$ & $\Theta(\log N)$~\cite{frieze2015recursive} & $\Theta(\sqrt{N})$~\cite{bollobas2001degree, mori2005maximum} &
  $\Theta(N)^\ddagger$~\cite{krapivsky2000connectivity} & $N-1$ \\
& leaves & $2$ & $\nicefrac{1}{2}N$~\cite{frieze2015recursive} & $\nicefrac{2}{3}N$~\cite{bollobas2001degree} &
  $\Theta(N)^\ddagger$~\cite{krapivsky2000connectivity} & $N-1$ \\
& diameter & $N-1$ & $\Theta(\log N)$~\cite{frieze2015recursive} & $\Theta(\log N)$~\cite{bollobas2004diameter} &
  $\Theta(1)^\ddagger$~\cite{krapivsky2000connectivity} & $2$ \\
\hline
\hline
\end{tblr}
}
\label{tab:results}
\end{center}
\end{table*}

\begin{figure}
    \centering
    \includegraphics[width=\columnwidth]{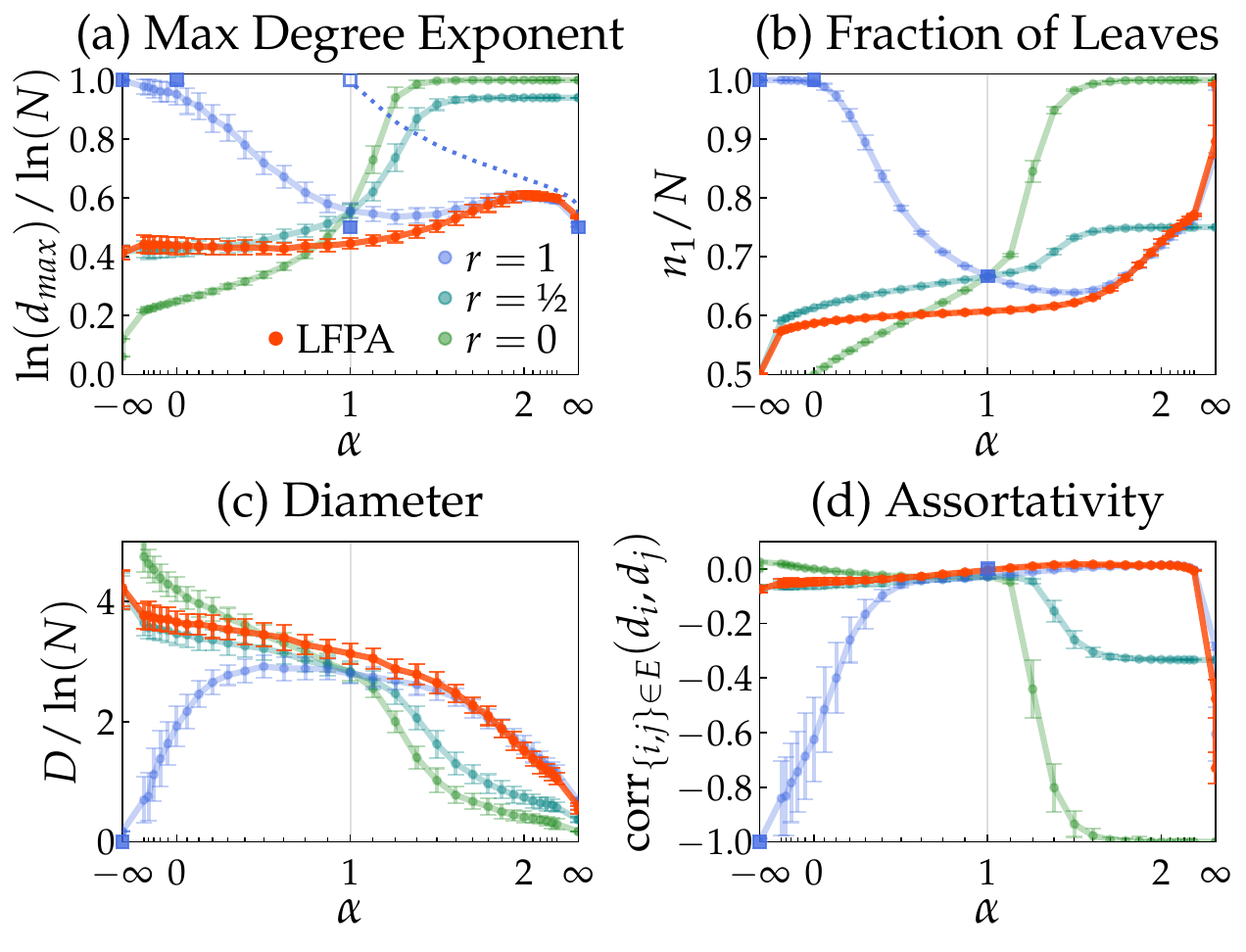}
     \vspace{-8mm}
\caption{\textbf{Simulation results of network observables.} Panels show (a) the (effective) maximum-degree exponent, (b) the fraction of leaves, (c) the diameter, defined as the longest distance between any two nodes in the network, and  (d) the assortativity coefficient~\cite{newman2002assortative}, for LFPA and CR networks grown to $10^5$ nodes starting from three nodes in a chain.  Theoretical predictions for CR at $r=1$ in the infinite-size limit are plotted as blue squares and a dotted curve. 
Each data point reports the mean and standard deviation over 100  simulations, visualized as a function of $\tanh(\alpha - 1)$. For finite-size scalings of these results, see Fig.~\ref{fig:FSS}. \hfill\hfill}
\label{fig:num_obs}
\end{figure}

\begin{figure}[t]
\centering
\includegraphics[width=0.9\columnwidth]{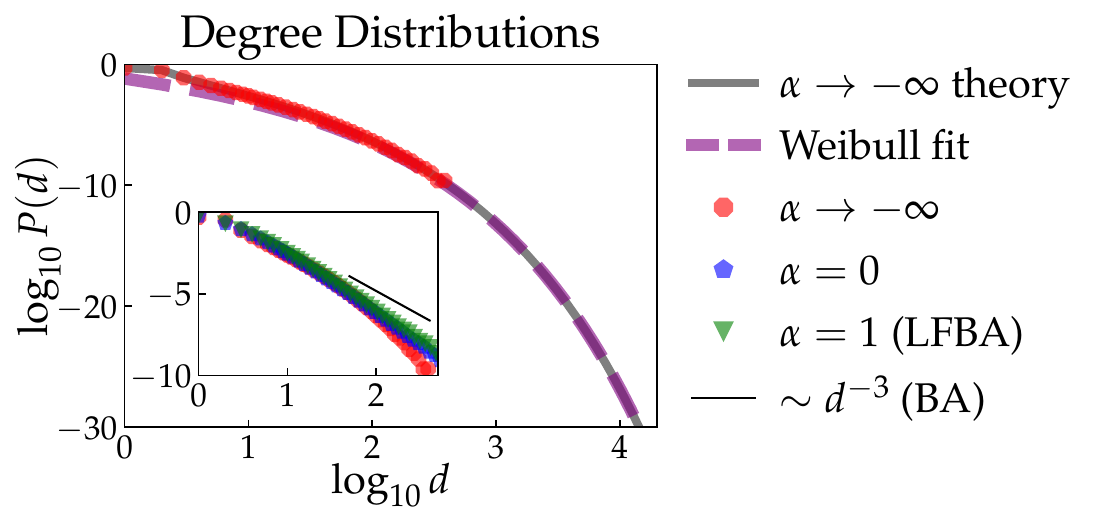}
\vspace{-2mm}
\caption{\textbf{Degree distributions of the (sub)linear LFPA models.}
Network simulations reveal a lighter tail than in the BA model for all $\alpha \leq 1$. 
In the $\alpha \to -\infty$ limit, we derive the exact recurrence Eq.~\eqref{eq_master} for the degree distribution, which yields the displayed theory curve. 
We further fit this theory curve with a Weibull stretched exponential, in line with our proof that it decays faster than any power-law (Theorem~\ref{thm:master_not_sf}). 
Networks are grown to  $10^6$  nodes starting from three nodes in a chain, and results are averaged over 500  simulations.\hfill\hfill}
\label{fig:num_deg}
\end{figure}

As an invariant starting point, the linear, $\alpha = 1$, case of the CR model is exactly the BA tree for any choice of $r$---as 
$
w_i^{\text{CR}}(1,r)
= (1-r)d_i + r \sum_{j \sim i} 1
= d_i
$
---featuring its characteristic scale-free behavior, namely a stationary power-law degree distribution with degree exponent $\gamma = 3$. In contrast, the LFPA model at $\alpha = 1$, which we refer to as the locally flexible BA (LFBA) model, gives rise to qualitatively different network topology, as the BA model is sensitive to heterogeneous redirection.
We prove that LFBA has a strictly smaller leaf proportion than BA (Theorem~\ref{thm:lfba_not_ba}), see also Fig.~\ref{fig:num_obs}b. 
Moreover, we show that LFBA growth cannot maintain the large hubs of the BA tree (Theorem~\ref{thm:BA_unstable_polynomial}). In the asymptotic absence of degree correlations, this statement can be further generalized to any scale-free degree distribution with exponent $2<\gamma < \nicefrac{11}{3}$ (Theorem~\ref{thm:uncorrelated_scalefree_unstable}). 
Under LFBA, high-degree nodes grow slower than in BA because, for any edge $(i,j)$ with $d_j < d_i$, redirection across the edge favors $j$ since $d_i/(d_i+1) > d_j/(d_j+1)$.
Network simulations indeed confirm a slower maximum-degree growth (Fig.~\ref{fig:num_obs}a and Fig.~\ref{fig:small_alpha_master}f) and a lighter-tailed degree distribution (Fig.~\ref{fig:num_deg} and Fig.~\ref{fig:small_alpha_master}c). Based on these results, we conjecture a stationary power-law degree distribution for LFBA with $\gamma \ge \nicefrac{11}{3} \approx 3.67$, close to the numerical estimate $\gamma \approx 3.7$.
Simulations also reveal that, at finite network sizes, LFBA exhibits slightly higher degree assortativity than BA, see Fig.~\ref{fig:num_obs}d. Note that the BA assortativity coefficient is negative at finite sizes but vanishes in the asymptotic limit~\cite{newman2002assortative, litvak2013uncovering}.

In the superlinear, $\alpha>1$, phase, we conjecture LFPA to be asymptotically equivalent to CR at $r=1$ with the same $\alpha$ value, because high-degree nodes, who are disproportionally selected as initial targets, redirect with asymptotic certainty under LFPA. Indeed, Fig.~\ref{fig:num_obs} shows that, in simulations, the LFPA observables closely track the CR at $r = 1$ values for $\alpha \gg 1$, but they deviate near $\alpha = 1$ due to pronounced finite-size effects.
Simulation results also indicate that the networks evolve into a hierarchical structure, with no stationary degree distribution, where connectivity is concentrated around dominant high-degree nodes and their close neighborhoods, while most nodes remain peripheral.
To approach the evolution of this structure analytically, we study the total graph weight dynamics of CR at $\alpha > 1$. We introduce the notion of a ``layered hierarchy'' to describe graphs organized into degree-based tiers and we find that a layered hierarchy with three essential tiers achieves the fastest progression of the total graph weight $W$, suggesting that the network should spontaneously evolve toward this structure.
Specifically, the three tiers consist of a degree-dominant node whose degree scales with network size $N$ as $N^{\nicefrac{\alpha}{(2\alpha - 1)}}$, a surrounding set of rich followers whose degrees scale as $N^{\nicefrac{(\alpha-1)}{(2\alpha - 1)}}$, and the remaining poor nodes whose degrees remain finite (Theorem~\ref{thm:layered_hierarchy}). 
We also find this hierarchical structure to be closely related to a remarkable form of topological balance in the network, whereby the combination 
$
\eta_i = {\sum_{j \sim i} d_j^{\alpha - 1}}/{d_i}
$
is asymptotically independent of node~$i$ (Remark~\ref{remark:balance}). Therefore, simultaneously with the emergence of a strong degree heterogeneity, a notion of homogeneity emerges where all nodes tend toward having the same $\eta$ value.
Note that our hierarchy results encompass the entire $\alpha>1$ phase of LFPA and the CR model at $r=1$, revealing behavior that greatly extends the qualitatively similar behavior found in the ``$k2$'' model, i.e.~CR at $\alpha=2, r=1$, ~\cite{falkenberg2020identifying, lee2023emergence}.
The limiting case $\alpha \to \infty$ yields a novel ``rich club'' network architecture. Here, the degree gap between the dominant node and the rich followers vanishes, giving rise to $\Theta(\sqrt{N})$ hubs, each of degree $\Theta(\sqrt{N})$ (Theorem~\ref{thm:largest_degree_alpha=infty_r=1}), with the role of the highest-degree node changing infinitely many times (Theorem~\ref{thm:infinitely_many_rivals_r=1}). 
Finally, both LFPA and CR in this superlinear phase display the small-world property~\cite{watts1998collective}, with diameters that scale (at most) logarithmically with network size, see Fig.~\ref{fig:num_obs}c.
Moreover, as the diameter normalized by $\log N$ is found to be decreasing for larger network sizes (Fig.~\ref{fig:FSS}c), we conjecture that the diameter saturates at a finite value for all $\alpha>1$.

In the entire (sub)linear, $\alpha \leq 1$, phase of the LFPA model, we observe  stationary degree distributions with lighter tails than BA (Fig.~\ref{fig:num_deg}).
In the limiting case, $\alpha \to -\infty$, the LFPA model is identical to the CR model with $\alpha \to -\infty, r = \nicefrac{1}{2}$, since an incoming node always selects a leaf and is then redirected with probability $\nicefrac{1}{2}$.
Here, we provide an exact recurrence relation fully capturing the degree distribution (Theorem~\ref{thm:master_equation}): let $q_{k,\ell}$ denote the fraction of nodes with $k$ leaf neighbors and $\ell$ non-leaf neighbors, then
\begin{align} \label{eq_master}
(2k+1)q_{k, \ell}=(k-1)q_{k-1, \ell}+(k+1)q_{k+1, \ell-1}
\end{align}
with initial conditions $q_{0, 1}=\nicefrac{1}{2}, q_{1, 1}=\nicefrac{1}{6}$.
Based on this recurrence, we prove that the tail of the degree distribution must decay faster than any power-law, as every moment of the asymptotic degree distribution is finite (Theorem~\ref{thm:master_not_sf}). On the contrary, a scale-free tail would lead to the divergence of all moments of order $\geq \gamma - 1$.
Intuitively, this result can be conceptualized through a notion of ``effective retirement'', where high-degree hubs struggle to acquire new neighbors as the number of their leaf neighbors diminishes, leading to a decreasing (effective) maximum-degree exponent for larger network sizes, see Fig.~\ref{fig:FSS}a.
Although a complete solution of the recurrence Eq.~\eqref{eq_master} appears to be intractable, a marginal recurrence of $k$ is analytically solvable~\cite{hartle2026leaves}, see Remark~\ref{remark:leafdegree}. Additionally, we obtain the exact asymptotic degree distribution by iterating this recurrence numerically, matching the network simulation results.
Motivated by the maximum degree as an extreme value statistic of the degree distribution, we fit the recurrence-iterated degree distribution with a Weibull-form stretched exponential,
$
f(d)=(k/\lambda)(d/\lambda)^{k-1}e^{-(d/\lambda)^k},
$
obtaining parameters $k \approx 0.335$ and $\lambda \approx 0.063$, see Fig.~\ref{fig:num_deg}.
A Kolmogorov--Smirnov test with bootstrapping from the fitted distribution finds no significant fit inconsistency at significance level $0.01$, over $>20$ orders of magnitude in degree frequency.
For a convergence analysis of the fitted Weibull parameters as functions of tail cutoff in the degree distribution, see Fig.~\ref{fig:weibull_trend}.
The Weibull-like tail implies that the maximum degree scales as $\Theta((\log N)^{\nicefrac{1}{k}})$, in line with Fig.~\ref{fig:small_alpha_master}g.
At finite $\alpha$, we cannot rule out a scale-free tail. Based on Fig.~\ref{fig:small_alpha_master}, we conjecture that degree distributions remain scale-free for finite $\alpha \le 1$, with an exponent $\gamma \ge \nicefrac{11}{3}$ that increases slowly as $\alpha$ decreases from the $\alpha = 1$ endpoint, eventually diverging as $\alpha \to -\infty$.
Compared to the superlinear phase, LFPA in this (sub)linear phase has a smaller, yet still dominant, fraction of leaves that monotonically decreases as $\alpha$ decreases (Fig.~\ref{fig:num_obs}b), as well as a larger, yet still logarithmically scaling diameter, retaining the small-world property (Fig.~\ref{fig:num_obs}c). 

Beyond facilitating the analysis of LFPA, the CR model family is interesting in its own right.
Our results on the rich emergent behavior of CR are summarized in Table~\ref{tab:results} alongside LFPA.
For all $\alpha \neq 1$, introducing any positive redirection probability $r > 0$ leads to qualitatively different network topology. 
In particular, as seen in Fig.~\ref{fig:num_obs}, unlike for nonlinear preferential attachment at $r=0$, all studied observables showcase non-monotonic dependence on $\alpha$ at $r=1$.
In general, the two-dimensional phase diagram unifies distinct cases of the CR model family into a continuum, not only revealing newly identified structures but also relating these models to one another.
For example, one way to approach the rich club model is as a limiting case of the $\alpha > 0, r = 1$ hierarchical regime. However, it is revealing to approach it as a limiting case of star-like architectures in $\alpha \to \infty, r \in (0, 1)$. We prove that these star-like architectures showcase a stabilized highest-degree node and finite diameter (Theorem~\ref{thm:king_alpha_infty}), and we solve their degree distributions precisely (Theorem~\ref{thm:leaves_proportion}). 
Complex network architectures arise also for intermediate values of $\alpha$ and $r$, e.g.~the $\alpha = 0, r = \nicefrac{1}{2}$ model leads to a network topology that appears to be qualitatively similar to LFBA, see Fig.~\ref{fig:degdist_comparison}.

\emph{Discussion.---}
We demonstrated that even a minimal amount of local flexibility in establishing connections renders the standard nonlinear preferential attachment inadequate for describing the resulting network architectures.
In the superlinear $\alpha > 1$ phase, as the growth process disproportionately targets high-degree nodes, new nodes typically attach to either the highest-degree node or one of its neighbors. 
These dynamics result in a non-stationary hierarchical architecture centered around a small fraction of highly connected nodes, potentially forming a rich club.
Although we focus here on graph neighborhood with $R = 1$, we expect the main qualitative findings to persist for larger graph neighborhoods $R>1$, where larger $R$ leads to $2R+1$ layers in the hierarchy. 

In the $\alpha \leq 1$ phase, the growth process is biased toward the small-degree periphery, resulting in stationary degree distributions but with tails strictly lighter than that of the BA model and, eventually, a complete breakdown of scale-freeness as $\alpha\to -\infty$.
As a simpler manifestation of the BA model's sensitivity to heterogeneous redirection, we find that even a static $r_i$ drawn, e.g., uniformly randomly from $[0,1]$ alters the BA behavior. Note that, in stark contrast to a lighter tail in the presence of the degree-dependent redirection in the LFBA, static heterogeneity leads to a heavier tail with $\gamma \approx 2.4$, in line with our analytic expectations, see Fig.~\ref{fig:linear_pa_redirection} for details.

Local flexibility not only reshapes network structure, but also has tangible functional implications, impacting robustness to attacks, as well as spreading and communication dynamics on the network~\cite{pastor2001epidemic}.
For example, Fig.~\ref{fig:robustness} suggests that LFBA has an enhanced tolerance against targeted node attacks, meaning that it better preserves graph connectivity than BA when high-degree nodes fail~\cite{albert2000error}, potentially due to its higher assortativity~\cite{newman2002assortative, li2011emergence}.
Yet, as a minimal model, LFPA is not intended to capture all aspects of real-world networks in which local flexibility plays a role---instead, it is meant to serve as a starting point for more refined models.
For example, in quantum communication networks, the attachment probability within the target’s graph neighborhood may depend on node availability, reliability, bandwidth, geographic distance, and cost, to name a few important factors.
Similarly, local flexibility might arise through a diffusion process~\cite{saramaki2004scale, evans2005scale, cannings2013random, krapivsky2017emergent, berry2024random}, albeit with an attachment weight that decays with the distance from the initial target.
Understanding the extent to which local attachment flexibility factors into the evolution of real-world networks is a rich direction for future work.

\emph{Code availability.---} Python code for efficient LFPA and CR network simulations is available at \url{https://zenodo.org/records/22314855}.

\emph{Acknowledgments.---}We thank X.~Meng and L.~Luzzatto for helpful discussions.
This work benefited from the 2023 and 2024 ``Focused Workshops on Networks and Their Limits'' held at the Erdős Center in Budapest, Hungary, supported by the ERC Synergy Grant DYNASNET (810115).
T.Z.~and I.A.K.~acknowledge support from the NSF–Simons National Institute for Theory and Mathematics in Biology, jointly funded by the U.S. National Science Foundation (Award No.~2235451) and the Simons Foundation (Award No.~MP-TMPS-00005320).
I.A.K.~was also supported by the National Science Foundation of the QIS program in the Division of Physics (Grant No.~PHY-2310706).
B.M.~was supported by the NKFIH (Grant No.~152535), funded by the Ministry of Innovation and Technology of Hungary from the National Research, Development and Innovation Fund.
G.Ó.~was supported by the Swiss National Science Foundation (Grant No. P500PT-211129) and by the Austrian Science Fund Cluster of Excellence (Grant No.~10.55776/COE7).

\bibliography{refs}

\clearpage
\appendix


\section{Analytic Results for the CR Model}

\begin{theorem}\label{thm:king_alpha_infty}
For $\alpha \to \infty$ and $r < 1$, the highest-degree node stabilizes almost surely, hence gaining degree at linear rate $1-r$. Consequently, the diameter of the tree is finite almost surely.
\end{theorem}

\begin{proof}[Proof sketch]
We analyze \(D_1(N)\) and \(D_2(N)\), the highest and second-highest degrees at size \(N\). We show that \(D_1>D_2+2\) occurs infinitely often almost surely, and that whenever this gap holds, the probability that a new highest-degree node ever emerges is \(O(D_1^{-2})\). Borel--Cantelli~\cite{klenke2008probability} then implies that new highest-degree nodes emerge only finitely many times almost surely.
\end{proof}

\begin{theorem} \label{thm:leaves_proportion}
For $\alpha \to \infty$ and $r \in (0, 1)$, 
\begin{align*}
p_k \to
\begin{cases}
1 - r + r^2, & \text{if $k = 1$,}\\[4pt]
(1-r)^2 r^{\,k-1}, & \text{if $k \ge 2$,}
\end{cases}
\end{align*}
where $p_k$ denotes the proportion of degree-$k$ nodes, excluding the highest-degree node.
\end{theorem}

\begin{proof}[Proof sketch]
By Theorem~\ref{thm:king_alpha_infty}, new nodes either attach directly to the single persistent highest-degree node, becoming its neighbors, or are redirected to one of its existing neighbors, increasing that neighbor's degree. This reduces the degree distribution to a growth process along the neighborhood of the highest-degree node, from which the stated limiting proportions follow.
\end{proof}

\begin{theorem} \label{thm:largest_degree_alpha=infty_r=1}
In the rich club model at $\alpha \to \infty$ and $r = 1$, the largest degree $D_1(N)$ satisfies
\begin{align*}
D_1(N) = \Theta(\sqrt{N}) \quad \text{w.p.\ 1.}
\end{align*}
\end{theorem}

\begin{proof}[Proof sketch]
The process alternates between periods with a unique highest-degree node and periods in which several adjacent nodes are tied for the highest degree. The largest degree can increase only when such a tie is resolved, and the time needed for it to grow from $d$ to $2d$ is $\Theta(d^2)$. The result then follows from successive application of this bound.
\end{proof}

\begin{theorem} \label{thm:infinitely_many_rivals_r=1}
    In the rich club model at $\alpha \to \infty$ and $r = 1$, infinitely many nodes are involved in a tie for the highest degree infinitely many times.
\end{theorem}

\begin{proof}[Proof sketch]
Assume instead that only a fixed finite set of nodes can tie for the highest degree infinitely often. Comparing a boundary node in this set with a neighboring node outside it, their degree difference evolves symmetrically until an outside node joins the tie, so it returns to zero infinitely often. Each return gives the outside node a positive chance to become involved in a tie for the highest degree, contradicting the assumption.
\end{proof}

\begin{theorem}[Layered hierarchy] \label{thm:layered_hierarchy}
We say that $G_N$ has a \emph{layered hierarchy} with exponents $\beta_1, \beta_2, \dots, \beta_k$ if there is a single node with degree $\asymp N^{\beta_1}$ forming the first layer, and each node in the $i$-th layer has degree $\asymp N^{\beta_i}$. The $(k+1)$-th layer consists of leaves. As the total number of nodes of $G_N$ is $N$, we have $N^{\beta_1} N^{\beta_2} \cdots N^{\beta_k} \asymp N $, so $\sum_{i=1}^{k} \beta_i = 1$. 

At $\alpha > 1$ and $r = 1$,
among all graphs with a layered hierarchy, expected total graph weight growth $\mathbb{E}[W(G_{N+1}) - W(G_N) \mid G_N]$ is maximized by those with three essential layers, i.e., those for which
\begin{align*}
\beta_1 = \frac{\alpha}{2\alpha - 1}, 
\quad 
\beta_2 = \frac{\alpha - 1}{2\alpha - 1}.
\end{align*}
\end{theorem}

\begin{proof}[Proof sketch]
We can write the expected one-step total graph weight growth as $\gamma_\alpha(N)/m_\alpha(N)$, where $m_\alpha(N)=W(G_N)=\sum_{i}d_i^\alpha$ and, to leading order, $\gamma_\alpha(N)=2\alpha\sum_{i\sim j}(d_i d_j)^{\alpha-1}$. Expressing these two quantities in terms of the layer exponents turns the problem into an optimization over the $\beta_i$. This optimization is maximized when only two non-leaf exponents are nonzero, and the constraint $\sum_i\beta_i=1$ gives the stated values.
\end{proof}

\begin{remark}[Balance property] \label{remark:balance}
The optimal hierarchy in Theorem~\ref{thm:layered_hierarchy} induces the following balance property: 
\begin{align*}
\eta_i \equiv \frac{\sum_{j \sim i} d_j^{\alpha-1}}{d_i}
\end{align*}
is asymptotically independent of node $i$. We argue that this balance and the maximization goal of Theorem~\ref{thm:layered_hierarchy} are symbiotic: if the new leaf attaches to node $i$, then the changes of $\gamma_\alpha, m_\alpha$ are given by
\begin{align*}
\Delta \gamma_\alpha &= (d_i + 1)^{\alpha-1}
+ \left((d_i + 1)^{\alpha-1} - d_i^{\alpha-1}\right)
\sum_{j \sim i} d_j^{\alpha-1} \\ &\asymp d_i^{\alpha-2} \sum_{j \sim i} d_j^{\alpha-1},
\\ \Delta m_\alpha
&= 1 + (d_i + 1)^{\alpha} - d_i^\alpha
\asymp d_i^{\alpha-1},
\end{align*}
and hence
$
{\Delta \gamma_\alpha}/{\Delta m_\alpha}
\asymp \eta_i.
$ To optimize this, it is advantageous to attach to nodes with large $\eta_i$. Such attachments reduce the value of $\eta_i$, so that in the optimal configuration, the quantities $\eta_i$ should be asymptotically constant among all nodes. In the BA limit, the balance property is trivially satisfied for every network structure, while for $\alpha>1$, the balance property itself implies a regular hierarchy: all leaves are attached to non-leaves of the same degree, which in turn have precisely one non-leaf neighbor with the same degree, and so on.
\end{remark}

\begin{theorem} \label{thm:leaves_alpha=-infty}
If $\alpha \to -\infty$, the proportion of leaves $p_1(N)$ tends to $r$ almost surely, and for $r < 1$ the diameter of the network goes to infinity almost surely.
\end{theorem}

\begin{proof}[Proof sketch]
In this limit, each new node first targets a leaf. The number of leaves increases exactly when redirection occurs, which happens with probability $r$, so $p_1(N)\to r$ by the law of large numbers. For $r<1$, attachments that increase the diameter continue to occur infinitely often, forcing the diameter to diverge.
\end{proof}

\begin{theorem} \label{thm:leaves_alpha=0}
    If $\alpha = 0$, $r \neq 1$, then for the asymptotic proportion of leaves
    \begin{align*}
    \frac{1}{2}\le \liminf_{N \to \infty} p_1(N), \ \limsup_{N\to\infty} p_1(N)\leq \frac{1}{2-r} .
    \end{align*}
    These bounds extend to $\alpha>0$: the lower one for any $r\in [0, \nicefrac{1}{2}]$, while the upper one for any $r\in [0, 1]$.
\end{theorem}

\begin{proof}[Proof sketch]
We track the probability that the next attachment is to a non-leaf, since this is when a new leaf is created. Bounding this probability in terms of the current leaf proportion and passing to the asymptotic limit gives the stated inequalities.
\end{proof}

\section{Analytic Results for the LFPA Model}

\begin{theorem} \label{thm:lfba_not_ba}
The LFPA model at $\alpha = 1$, i.e. the LFBA model, has strictly fewer leaves than the BA model in the long run.
\end{theorem}

\begin{proof}[Proof sketch]
We show that $w^{\text{LFBA}}_{\text{leaf}} \ge \nicefrac{7}{6}$, in contrast to $w^{\text{BA}}_{\text{leaf}} = 1$. Attaching to a leaf is equivalent to increasing the number of non-leaves by one, and vice versa. By analyzing the asymptotic limit, we show that the fraction of leaves in LFBA is bounded above by $\nicefrac{12}{19}\approx0.63$, whereas in BA it is known to be $\nicefrac{2}{3}$~\cite{bollobas2001degree}.
\end{proof}

\begin{theorem}
\label{thm:BA_unstable_polynomial}
The hub structure of the BA tree, having precise adherence to the BA degree correlation~\cite{fotouhi2013degree}, is unstable under the LFBA dynamics.
\end{theorem}

\begin{proof}[Proof sketch]
In the LFBA tree, let $i_N$ be a highest-degree node at size $N$, and let $D_N:=d(i_N)$. Assume, for contradiction, that the local structure around $i_N$ is asymptotically BA, with the neighbors of the hub following the BA conditional degree distribution~\cite{fotouhi2013degree}.
We show that for every $\varepsilon>0$, the induced LFBA hub growth satisfies
\begin{align*}
    D_N
    =
    O \left( N^{\frac12\left(\frac{\pi^2}{6}-1\right)+\varepsilon} \right)
    \approx O \left( N^{0.322 + \varepsilon} \right)
\end{align*}
for all sufficiently large $N$, contradicting the BA hub law $D_N \asymp N^{\nicefrac{1}{2}}$~\cite{bollobas2001degree}.
\end{proof}

\begin{theorem}
\label{thm:uncorrelated_scalefree_unstable}
The hub structure of any degree-uncorrelated scale-free tree with degree exponent $2<\gamma<\nicefrac{11}{3}$, and in particular $2<\gamma\le 3$, is unstable under the LFBA dynamics.
\end{theorem}

\begin{proof}[Proof sketch]
Assume, for contradiction, that the LFBA tree is asymptotically
scale-free with no degree correlations, i.e.
$
    p(k)\sim k^{-\gamma}, \ p(\ell\mid k) \sim \ell p(\ell).
$
We show that at all sufficiently large $N$, $D_N = O(N^{\nicefrac38+\varepsilon})$ for any $\varepsilon > 0$. However, having $2<\gamma<\nicefrac{11}{3}$ implies $D_N\asymp N^{\beta}$ where $\beta >\nicefrac38$, giving a contradiction.
\end{proof}

\begin{theorem}
\label{thm:master_equation}
Consider LFPA at $\alpha \to -\infty$ (or equivalently, CR at $\alpha \to -\infty$, $r = \nicefrac{1}{2}$). The asymptotic degree distribution is given by
$
p_x = \sum_{k+\ell = x} q_{k,\ell}
$, 
where $q_{k,\ell}$ denotes the fraction of nodes with $k$ leaf neighbors and $\ell$ non-leaf neighbors. The values $q_{k,\ell}$ are determined by the stationary linear recurrence
\begin{align*} \tag{\ref{eq_master}}
(2k+1)q_{k, \ell}=(k-1)q_{k-1, \ell}+(k+1)q_{k+1, \ell-1}
\end{align*}
with initial conditions $q_{0, 1}=\nicefrac{1}{2}, q_{1, 1}=\nicefrac{1}{6}$.
\end{theorem}

\begin{proof}[Proof sketch]
We write down a one-step master equation for $Q_N(k,\ell)$, the expected number of nodes in state $(k,\ell)$ when the tree has $N$ nodes, by enumerating the possible transitions caused by the chosen pivot leaf. This master equation provides the basis to derive Eq.~\eqref{eq_master} in the asymptotic limit.
\end{proof}

\begin{theorem} \label{thm:master_not_sf}
    For $\alpha\to -\infty$, every moment of the asymptotic degree distribution of the LFPA model is finite. In particular, LFPA at $\alpha\to -\infty$ is not scale-free. 
\end{theorem}

\begin{proof}[Proof sketch]
With the notation of Theorem~\ref{thm:master_equation}, for the $s$-th moment our claim is equivalent to
\begin{align*} 
\sum_{k\geq 0, \ell\geq 1}(k+\ell)^sq_{k, \ell}<\infty.
\end{align*}
Taking the recurrence Eq.~\eqref{eq_master} with weight $\phi(k, l)$ and adding up the resulting equations, after some formal manipulations we obtain a convergent series involving $q_{k, l}$. By certain linear algebraic observations, we can guarantee the existence of $\phi(k, l)$ for which this convergent series precisely coincides with the left-hand side of the above inequality. The formal manipulations are made rigorous by a truncation argument.
\end{proof}

\begin{remark}
\label{remark:leafdegree}
A concurrent work by Hartle and Krapivsky~\cite{hartle2026leaves} provides further analysis of Eq.~\eqref{eq_master}.
They note, which we agree, that solving this two-index recurrence to obtain the full degree distribution is difficult. However, by marginalizing over $\ell$, they formulate and solve a closed recurrence for $k$ alone, obtaining a ``leaf-degree distribution'' with a stretched exponential tail.
\end{remark}

\nocite{BorkarMeyn}

\end{document}